\documentclass[a4paper]{amsproc}
\usepackage{amssymb}
\usepackage{amscd}
\usepackage{pstricks}
\usepackage[dvips]{graphicx}\usepackage{pst-3dplot}
\usepackage[all,cmtip]{xy}

\newtheorem{thm}{Theorem}
\newtheorem{prop}[thm]{Proposition}
\newtheorem{lem}[thm]{Lemma}

\newtheorem{remark}{Remark}

\newcommand{\R}{\mathbb{ R}}

\DeclareMathOperator{\pr}{pr} 
\DeclareMathOperator{\diag}{diag} \DeclareMathOperator{\ad}{ad}
\DeclareMathOperator{\Ad}{Ad} \DeclareMathOperator{\tr}{tr}

\title[Rolling balls over spheres in $\R^n$]{ROLLING BALLS OVER SPHERES IN $\R^n $}


\author[Bo\v zidar Jovanovi\'c]{}

\email{bozaj@mi.sanu.ac.rs}

\begin{document}

\maketitle

\centerline{\scshape Bo\v zidar Jovanovi\'c}
\medskip
{\footnotesize
\centerline{Mathematical Institute SANU}
\centerline{Serbian Academy of Sciences and Arts}
\centerline{Kneza Mihaila 36, 11000 Belgrade, Serbia}
}

\begin{abstract}
We study the rolling of the Chaplygin ball in $\R^n$  over a fixed $(n-1)$--dimensional sphere
without slipping and without slipping and twisting.
The problems can be naturally considered within a framework of
appropriate modifications of the L+R and LR systems -- well
known systems on Lie groups groups with an invariant measure.
In the case of the rolling without slipping and twisting,
we describe the $SO(n)$-Chaplygin reduction to $S^{n-1}$ and prove the Hamiltonization of the reduced system for a special inertia operator.
\end{abstract}

\tableofcontents

\section{Introduction}

Let $(Q,L,\mathcal D)$ be a nonholonomic Lagrangian system, where
$Q$ is a $n$-dimensional manifold, $L: TQ\to \R$ Lagrangian, and $\mathcal D$ nonintegrable
$(n-k)$-dimensional distribution of constraints.  Let
$q=(q_1,\dots,q_n)$ be some local coordinates on $Q$ in which the
constraints are written in the form
\begin{equation}
\sum_{i=1}^n \alpha_i^j(q)\dot q_i=0,\qquad
j=1,\dots,k. \label{constraints}\end{equation}
The motion of  the system
is described by the Lagrange-d'Alembert equations
\begin{equation}
\frac{d}{dt}\frac{\partial L}{\partial \dot q_i}=\frac{\partial
L}{\partial q_i}+ \sum_{j=1}^{k}\lambda_j \alpha^j_i,  \qquad
i=1,\dots,n,\label{Hamilton}
\end{equation}
where the Lagrange multipliers $\lambda_j$ are chosen such that the
solutions $q(t)$ satisfy constraints \eqref{constraints}.
The sum $\sum_{j=1}^{k}\lambda_j \alpha^j_i$
represents the reaction force of the constraints.

The nonholonomic systems, generically, are not Hamiltonian
systems. However, many constructions from the theory of
Hamiltonian systems, such as Noether's theorem and the reduction
of symmetries, apply with certain modifications (e.g, see
\cite{AKN, BKMM, C, CCLM, FS, Jo7, Jo8, Koi, St}). Besides, some systems have an
invariant measure, which puts them rather close to Hamiltonian
systems and allow the integration using the Euler--Jacobi theorem
(e.g., see \cite{AKN}).

The existence of invariant measure for
various nonholomic problems is studied extensively (e.g., see
\cite{FNM, Gr, Jo1, Ko, KO,  ZB}).
The LR systems introduced by Veselov and Veselova \cite{VeVe1, VeVe2} and L+R systems introduced by Kozlov and Fedorov \cite{FeRCD, FeKo} on unimodular Lie groups are one of the basic and remarkable
examples.

The closely related problem is
the Hamiltonization of nonholonomic systems, in particular, after the time reparametrisation by using the Chaplygin reducing multiplier
(e.g., see \cite{BN, BBM, BBM3, BM, CCLM, EKR, FeJo, H, OFBZ, Tat}). In the case of integrability, the dynamics over
regular invariant $m$--dimensional tori, in the original time,
has the form
\begin{equation}
\dot\varphi_1=\omega_1/\Phi(\varphi_1,\dots,\varphi_m),\dots,
\dot\varphi_m=\omega_m/\Phi(\varphi_1,\dots,\varphi_m), \qquad \Phi>0.
\label{Jacobi2}
\end{equation}

Inspired by the study of the rolling of a of a balanced, dynamically asymmetric ball without slipping (after Chaplygin \cite{Ch1} usually called the \emph{Chaplygin ball} or the \emph{marble Chaplygin ball} \cite{EKR}) and without slipping and twisting (referred as the \emph{rubber Chaplygin ball} in \cite{EKR})
over a fixed sphere in $\R^3$, given by Borisov, Fedorov, and Mamaev \cite{BF,BM2,BFM,BM3} and Ehlers and Koiller \cite{EK},
we study the associated nonholonomic problems in $\R^n$:
the rolling of the Chaplygin ball in $\R^n$ over a fixed $(n-1)$--dimensional sphere without slipping (and twisting). The problems can be naturally considered within a framework of
appropriate modifications of the L+R and LR systems, recently introduced in \cite{Jo5}.

Note that $n$--dimensional nonholonomic rigid body problems:
the Veselova problem \cite{FeJo}, the Suslov problem \cite{FeJo2},
the rolling of the rubber Chaplygin ball \cite{Jo3} and the
Chaplygin ball \cite{Jo4} over hyperplane in $\R^n$ (at the zero
level set of the $SO(n-1)$--momentum mapping), for certain inertia
operators, are Hamiltonizable systems. Moreover, all mentioned
models are integrable as well, and a motion over a generic
invariant tori has the form \eqref{Jacobi2}\footnote{The Suslov
problem studied in \cite{FeJo2} is an exception. There, the
invariant manifolds not need to be tori.}. In this paper we prove
that the rolling of the rubber Chaplygin ball over a sphere allows
Chaplygin Hamiltonization, while, however, in general the problem
is not integrable.

For a given
nonintegrable distribution $\mathcal D$ on a Riemannian manifold
$Q$, there is an alternative, important, variational or
sub-Riemannian problem, that is already Hamiltonian. The
variational problem for rolling of a $(n-1)$--sphere on spaces of
constant curvature is studied by Jurdjevic and Zimmerman
\cite{JZ}.

\subsection{Result and outline of the paper}
In Section 2 we consider a motion of the Chaplygin ball
of radius $\rho$ without slipping (the
velocity of the contact point equals zero) over a fixed
sphere in $\R^n$ of radius $\sigma$ in three variants of the problem.
The first one represents the motion of the
ball over outside surface of the fixed sphere, the second one is the rolling over inside surface of the fixed sphere, and the third one is the case where Chaplygin ball represents spherical shell with fixed sphere placed in its interior.
The systems are described in Proposition \ref{SSR}.
In all cases the configuration space is $SO(n)\times S^{n-1}$ and the nonholonomic distribution is diffeomorphic to $TSO(n)\times S^{n-1}$.

It appears that these nonholonomic problems are examples of \emph{$\epsilon$-modified L+R systems} (see \cite{Jo5}) with the parameter
\begin{equation}\label{Epsilon}
\epsilon=\frac{\sigma}{\sigma\pm\rho},
\end{equation}
and we directly
obtain an invariant measure (see Theorem \ref{TC}, item (i)), which takes the simpler form for the inertia operator (Theorem \ref{TC}, item (ii))
\begin{equation}\label{ch-op}
\mathbb I(E_i\wedge E_j)=\frac{Da_ia_j}{D-a_ia_j}E_i\wedge E_j.
\end{equation}
Here $0< a_ia_j <D$, $i,j=1,\dots,n$, $E_1,\dots,E_n$ is the standard base of $\R^n$:
\begin{equation}\label{baza}
E_1=(1,0,\dots,0,0)^T,\dots, E_n=(0,0,\dots,0,1)^T,
\end{equation}
and $D=m\rho^2$, where $m$ and $\rho$ are the mass and the radius of the rolling ball, respectively.

The operator \eqref{ch-op} is introduced in \cite{Jo4} in the study of a related problem of rolling of the Chaplygin ball over a horizontal hyperplane in $\R^n$. Rolling over the horizontal plane can be seen as the limit case, where $\epsilon$ becomes 1, as the radius of the fixed sphere $\sigma$ tends to infinity.
Although we have the Hamiltonization of the system for $\epsilon=1$ (at the zero
level set of the $SO(n-1)$--momentum mapping), the Hamiltonization, and eventually integrability, for $\epsilon\ne 1$ is still an open problem.

In Section 3, we study the rolling with additional constraints determined by the non-twist condition of the ball at the contact point (the infinitesimal rotation of the ball in the tangent plane to the contact point are forbidden),
referred as the rubber Chaplygin ball problem.
The equations of motion are described in Proposition \ref{SSRR}. Now, the distribution of constraints is
$(n-1)$-dimensional and represents the connection of the principal bundle
\begin{equation}
\label{principal}
\xymatrix@R28pt@C28pt{
SO(n) \ar@{^{}->}[r]  & SO(n) \times S^{n-1} \ar@{^{}->}[d]^{\pi}  \\
  & S^{n-1}  }
\end{equation}
with respect  to the diagonal $SO(n)$-action, i.e., the system is a $SO(n)$-Chaplygin system.

We also consider an appropriate \emph{extended system}
allowing the integrals that replace the non-twist condition
(the rubber Chaplygin ball problem is its subsystem,  Subsection \ref{extended-system}). The obtained system is an example of \emph{$\epsilon$-modified LR system} (see \cite{Jo5}),
implying the form of an invariant measure described in Theorems \ref{IMRR} and \ref{RS2}.
In particular, for the inertia operator
\begin{equation}\label{spec-op}
\mathbb I(E_i\wedge E_j)=(a_ia_j-D)E_i\wedge
E_j,
\end{equation}
the invariant measure, as in the case of non-rubber rolling and the operator \eqref{ch-op}, significantly simplifies (see Theorem \ref{RS2}, item (ii)).

Further, in Section 4, we derive the curvature of the nonholonomic distribution (see Lemma \ref{curvature}), describe the $SO(n)$-Chaplygin reduction to $S^{n-1}$ (Theorem \ref{REDsym}), as well as the reduced invariant measure (Theorem \ref{RMIM}). Finally, we obtain the Hamiltonization of the reduced system defined by the inertia operator \eqref{spec-op} (Theorem \ref{main}, Section 5).

\section{Chaplygin ball in $\R^n$}

\subsection{Kinematics}
We consider the Chaplygin ball type problem of rolling without
slipping of an $n$-dimensional balanced ball, the mass center $C$
coincides with the geometrical center, of radius $\rho$ in several
nonholonomic models:\footnote{It would be also interesting to study a modified problem, where we assume that the ball rolls over a rotating $n$--dimensional sphere (for $n=3$, see \cite{BoMaKi, FS}). Rolling of a $n$--dimensional Chaplygin ball over a rotating horizontal plane is considered in \cite{FNS}.}

\begin{itemize}
\item[(i)] rolling over outer surface of the $(n-1)$-dimensional
fixed sphere of radius $\sigma$, Figure 1a;

\item[(ii)] rolling over inner surface of the $(n-1)$-dimensional
fixed sphere of radius $\sigma$ ($\sigma>\rho$), Figure 1b;

\item[(iii)] rolling over outer surface of the $(n-1)$-dimensional
fixed sphere  of radius $\sigma$, but the fixed sphere is within
the rolling ball ($\sigma<\rho$, in this case, the rolling ball is
actually a spherical shell), Figure 1c.
\end{itemize}

We suppose that the origin $O$ of $\R^n$  coincides with the center of the fixed sphere. The configuration space is
 the {direct product} of Lie
groups $SO(n)$ and $\R^n$, where $g\in SO(n)$ is the rotation
matrix of the sphere (mapping a frame attached to the body to the
space frame) and  $\mathbf r=\overrightarrow{OC}\in {\mathbb R}^n$
is the position vector of its center $C$ (in the space frame).
The vector $\mathbf r$ belongs to the $(n-1)$-dimensional constraint
hypersurface $\mathcal S$ defined by the holonomic constraint
\begin{eqnarray*}
(\mathbf r,\mathbf r)=(\sigma\pm\rho)^2
\end{eqnarray*}
i.e, $\mathcal S$ is a sphere $S^{n-1}$.\footnote{From now on, whenever we have a sign
$\pm$, we take "$+$" for the case (i) and "$-$" in the cases
(ii) and (iii).}

As usual, for a trajectory $(g(t),\mathbf r(t))$ we define angular
velocities of the ball in the moving and the fixed frame, and
the velocity of the center $C$ of the ball in the fixed frame by
$$
\omega=g^{-1}\dot g, \qquad \Omega=\dot g g^{-1}=\Ad_g(\omega), \quad \mathbf
V=\dot {\mathbf r}=\frac{d}{dt}\overrightarrow{OC},
$$
respectively.

Let $A$ be the point of the rolling ball at the point of contact.
The condition for the ball to roll without
slipping leads that the velocity of the contact point is equal to
zero in the fixed reference frame:\footnote{Through the paper, we
consider vectors in $\R^n$ as columns and $\Omega\mathbf\Gamma$ denotes
the usual matrix multiplication.  The Euclidean scalar product of
$x,y\in\R^n$ is simply $(x,y)=x^Ty$, while the wedge product is
$x\wedge y=x\otimes y-y\otimes x=xy^T-yx^T$.}
\begin{eqnarray}
&&\frac{d}{dt}\overrightarrow{OA}=\frac{d}{dt}\big(\overrightarrow{OC}+\overrightarrow{CA}\big)=\mathbf
V-\rho{\Omega}\mathbf\Gamma =0 \qquad (\text{the case (i)});\nonumber \\
&&\frac{d}{dt}\overrightarrow{OA}=\frac{d}{dt}\big(\overrightarrow{OC}+\overrightarrow{CA}\big)=\mathbf
V+\rho{\Omega}\mathbf\Gamma =0 \qquad (\text{the cases (ii) and
(iii)}),\label{ch-constr}
\end{eqnarray}
where $\mathbf\Gamma\in{\mathbb R}^{n}$ is the unit normal to the fixed sphere at the contact point directed outward, or, equivalently, the direction of the contact point in the fixed reference frame:
\begin{equation}\label{GAM}
\mathbf\Gamma=\frac{1}{\vert\overrightarrow{OA}\vert}\overrightarrow{OA}=\frac{1}{\sigma\pm\rho}\mathbf r.
\end{equation}

Therefore, the nonholonomic
distribution is
\[
\mathcal D^{\pm}=\{(\omega,\mathbf V,g,\mathbf r) \, \vert\,
\mathbf V=\pm \rho \Ad_g(\omega)\mathbf\Gamma=\pm \frac{\rho}{
\sigma\pm\rho} \Ad_g(\omega)\mathbf r \}.
\]
It is clear that $\mathcal D^{\pm}$ is diffeomorphic to the
product $TSO(n)\times S^{n-1}$.

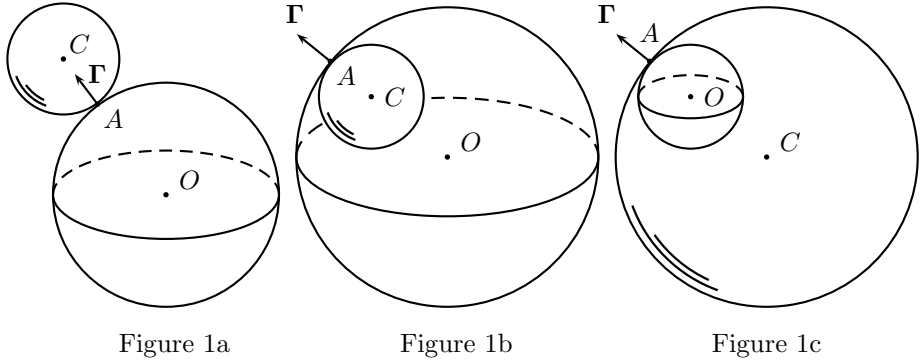
\begin{pspicture}(14,6.2)
\pscircle[linecolor=black](2.1,2.5){1.5}
\psellipticarc[linestyle=dashed](2.1,2.5)(1.5,0.6){0}{180}
\psellipticarc(2.1,2.5)(1.5,0.6){180}{360}
\psdot[dotsize=2pt](2.1,2.5)\uput[0](2.1,2.7){$O$}
\psdot[dotsize=2pt](1.2,3.7)\uput[0](1.1,3.5){$A$}
\psdot[dotsize=2pt](0.75,4.3)\uput[0](0.65,4.5){$C$}
\pscircle[linecolor=black](0.75,4.3){0.75}
\psarc[linewidth=0.03cm](0.75,4.3){0.66}{200}{250}
\psarc[linewidth=0.03cm](0.75,4.3){0.6}{215}{245}
\psline{->}(1.2,3.7)(0.9,4.1)
\uput[0](0.9,4.1){$\mathbf\Gamma$}
\pscircle[linecolor=black](5.8,3){2}
\psellipticarc[linestyle=dashed](5.8,3)(2,0.8){0}{105}
\psellipticarc[linestyle=dashed](5.8,3)(2,0.8){165}{180}
\psellipticarc(5.8,3)(2,0.8){180}{360}
\psdot[dotsize=2pt](5.8,3)\uput[0](5.8,3.2){$O$}
\psdot[dotsize=2pt](4.8,3.8)\uput[0](4.8,3.8){$C$}
\pscircle[linecolor=black](4.8,3.8){0.703}
\psarc[linewidth=0.03cm](4.8,3.8){0.61}{200}{250}
\psarc[linewidth=0.03cm](4.8,3.8){0.55}{215}{245}
\psline{<-}(3.82,4.63)(4.26,4.27)
\uput[0](3.5,4.9){$\mathbf\Gamma$}
\psdot[dotsize=2pt](4.26,4.27)
\uput[0](4.16,4.07){$A$}
\pscircle[linecolor=black](10,3){2}
\psdot[dotsize=2pt](10,3)\uput[0](10,3.2){$C$}
\psdot[dotsize=2pt](9,3.8)\uput[0](9,3.8){$O$}
\pscircle[linecolor=black](9,3.8){0.703}
\psellipticarc[linestyle=dashed](9,3.8)(0.703,0.3){0}{180}
\psellipticarc(9,3.8)(0.703,0.3){180}{360}
\psarc[linewidth=0.03cm](10,3){1.88}{200}{250}
\psarc[linewidth=0.03cm](10,3){1.80}{215}{245}
\psline{<-}(8.02,4.63)(8.46,4.27)
\psdot[dotsize=2pt](8.46,4.27)
\uput[0](7.6,4.9){$\mathbf\Gamma$}
\uput[0](8.15,4.65){$A$}
\uput[0](1.3,0.5){Figure 1a}
\uput[0](5,0.5){Figure 1b}
\uput[0](9.1,0.5){Figure 1c}
\end{pspicture}

\subsection{Dynamics in the fixed frame}
In what follows we identify $so(n)\cong so(n)^*$ by an invariant
scalar product
\begin{equation}
 \langle X,Y\rangle=-\frac12\mathrm{tr}(XY).
\label{KF}
\end{equation}

Let $m$ be the mass of the ball and $\mathbb I: so(n) \to so(n)^*\cong so(n)$ be the inertia tensor that defines a left--invariant metric on $SO(n)$. The Lagrangian of the system is
then given by
\begin{equation}
L(\omega,\mathbf V,g,\mathbf r)=\frac12\langle \mathbb
I\omega,\omega\rangle+\frac12 m(\mathbf V,\mathbf V),
\label{ch-lagr}
\end{equation}
where $(\cdot,\cdot)$ is the Euclidean scalar product in $\R^n$.

By the use of the
constraints \eqref{ch-constr} we find the form of reaction forces
in the right-trivialization of $SO(n)$ in which the equations
\eqref{Hamilton} become
\begin{align}
\dot M &=-(\pm\rho \mathbf\Lambda \wedge \mathbf\Gamma),  \label{c1}\\
m\dot{\mathbf V} &=\mathbf\Lambda, \label{c2} \\
\dot g &=\Omega\cdot g,\label{c3}\\
 \dot{\mathbf r} &=\mathbf V.     \label{c4}
\end{align}
where $M=\Ad_g(\mathbb I\omega)\in so(n)^*\cong so(n)$ is the ball
angular momentum in the fixed frame and  $\Lambda\in\R^n$ is the
Lagrange multiplier. Differentiating the constraints
(\ref{ch-constr}) and using (\ref{c2}),  we get
\begin{equation}\label{LAM}
\mathbf\Lambda= \pm m\rho(\dot\Omega\mathbf\Gamma + \Omega \dot{\mathbf\Gamma}).
\end{equation}

Further, \eqref{GAM} and \eqref{ch-constr} imply that the vector $\mathbf\Gamma$ in the
fixed frame satisfies the equation:
\begin{equation}\label{Gamma}
\dot{\mathbf\Gamma}=\frac{1}{\sigma\pm\rho}\mathbf
V=\pm\frac{\rho}{\sigma\pm\rho}\Omega\mathbf\Gamma.
\end{equation}

Finally, from \eqref{LAM} and \eqref{Gamma} we get that \eqref{c1} takes the form
\begin{equation}
\dot M=-D\big(\dot\Omega\, \mathbf\Gamma \otimes \mathbf\Gamma + \mathbf\Gamma \otimes
\mathbf\Gamma
\,\dot\Omega\big)-D\big(\pm\frac{\rho}{\sigma\pm\rho}\big)\big(\Omega\,\Omega\,\mathbf\Gamma\otimes
\mathbf\Gamma - \mathbf\Gamma \otimes \mathbf\Gamma \,\Omega\,\Omega\big),
\label{ch*}\end{equation}
where $D=m\rho^2$.

\subsection{Dynamics in the body frame and reduction}
Both the Lagrangian $L$ and the distribution $\mathcal D^\pm$
are invariant with respect to the \emph{left} $SO(n)$-action
\begin{equation}\label{left}
a\cdot (\omega,\mathbf V,g,\mathbf r)=(\omega,a\mathbf
V,ag,a\mathbf r), \qquad a\in SO(n).
\end{equation}

Therefore, the system can be reduced to
$$
so(n)\times S^{n-1}\cong (TSO(n)\times S^{n-1})/SO(n)\cong
\mathcal D^{\pm}/SO(n).
$$

Note that the $SO(n)$--action defines the principal bundle \eqref{principal},
where the submersion $\pi$ is given by
\begin{equation}\label{submersion}
{\gamma}=\pi(g,\mathbf
r)=\frac{1}{\sigma\pm \rho} g^{-1}\mathbf r=g^{-1}\mathbf\Gamma,
\end{equation}
that is, a base point of $(g,\mathbf r)$ is $\gamma=g^{-1}\mathbf\Gamma$, the unit normal at the contact point to the fixed sphere (directed outward)
in the frame attached to the ball.

We can use $(g,\gamma)$
 instead of $(g,\mathbf r)$, for coordinates of a configuration space. Then the
$SO(n)$--action \eqref{left} takes the form:
\begin{equation}\label{left2}
a\cdot (\omega,\dot\gamma,g,\gamma)=(\omega,\dot\gamma,ag,\gamma).
\qquad a\in SO(n).
\end{equation}

From \eqref{Gamma}, we get the kinematic equation for $\gamma$
\[
\dot\gamma =\frac{d}{dt}\big(g^{-1}\big){\mathbf\Gamma}+g^{-1}\dot{\mathbf\Gamma}=-g^{-1}\dot g g^{-1}{\mathbf\Gamma} \pm\frac{\rho}{\sigma\pm\rho}g^{-1}\Omega{\mathbf\Gamma}=-\omega\gamma\pm \frac{\rho}{\sigma\pm\rho}\omega\gamma.
\]

By introducing parameter $\epsilon$ (see \eqref{Epsilon}),
we can write it as a modified
Poisson equation
\begin{equation}\label{PE*}
\dot\gamma=-\epsilon\omega\gamma.
\end{equation}

Let
\begin{equation} \label{K}
\mathbf k=\kappa (\omega)=\mathbb
I\omega+D( {\omega \,} \gamma\otimes\gamma+ \gamma \otimes
\gamma{\,\omega} )\in so(n)\cong so(n)^*
\end{equation}
be the angular momentum of the ball relative to the contact point (see \cite{FeKo}).

\begin{prop}\label{SSR}
{\rm (i)} The complete set of equations on $T^*SO(n)\times S^{n-1}$ in
variables $(\mathbf k,g,\gamma)$ is given by
\begin{equation}
\dot{\mathbf k}=[\mathbf k,\omega], \qquad \dot g=g\cdot\omega,
\qquad \dot\gamma=-\epsilon\omega\gamma.\label{sfera-sfera}
\end{equation}

{\rm (ii)} The reduction of the left $SO(n)$--symmetry \eqref{left2} gives a
system on $so(n)^*\times S^{n-1}$ defined by the equations
\begin{equation}
\dot{\mathbf k}=[\mathbf k,\omega],  \qquad
\dot\gamma=-\epsilon\omega\gamma.\label{sfera-sfera-red}
\end{equation}
\end{prop}

\begin{proof}
By applying the identities
\begin{equation*}\label{oO}
\dot\omega=\Ad_{g^{-1}}(\dot\Omega),  \qquad \mathbb I\dot\omega-[\mathbb I\omega,\omega]=
\Ad_{g^{-1}}(\frac{d}{dt}\big(\Ad_{g}(\mathbb I\omega)\big)=\Ad_{g^{-1}}(\dot M),
\end{equation*}
to \eqref{ch*}, in the left trivialization of
$SO(n)$ we obtain the equation:
\begin{align}
\nonumber \mathbb I\dot\omega-[\mathbb I\omega,\omega] =&-D\big(
{\dot \omega \,} \gamma\otimes\gamma+ \gamma \otimes
\gamma{\,\dot\omega}\big)\\
\label{druga}&-D\big(\pm\frac{\rho}{\sigma\pm\rho}\big)\big(\omega\,\omega\,\gamma\otimes
\gamma - \gamma \otimes \gamma \,\omega\,\omega\big)\\
\nonumber =& -D\big( {\dot \omega \,} \gamma\otimes\gamma+
\gamma \otimes \gamma{\,\dot\omega}\big)+D(1-\epsilon)[{\omega
\,} \gamma\otimes\gamma+ \gamma \otimes \gamma{\,\omega},\omega].
\end{align}

Next, from \eqref{PE*} we have
\begin{align}
\nonumber
\frac{d}{dt}(\omega\gamma\otimes\gamma+\gamma\otimes\gamma\omega)
=&\dot\omega\gamma\otimes\gamma+\gamma\otimes\gamma\dot\omega-\epsilon\omega\omega\gamma\otimes\gamma\\
\label{izvod}& +\epsilon
\omega\gamma\otimes\gamma\omega-\epsilon\omega\gamma\otimes\gamma\omega+\epsilon\gamma\otimes\gamma\omega\omega
\\
\nonumber =&
\dot\omega\gamma\otimes\gamma+\gamma\otimes\gamma\dot\omega+
\epsilon[\omega\gamma\otimes\gamma+\gamma\otimes\gamma\omega,\omega].
\end{align}

As a result, from \eqref{druga} and \eqref{izvod} we obtain:
\begin{align*}
\dot{\mathbf k} =& \mathbb
I\dot\omega+D(\dot\omega\gamma\otimes\gamma+\gamma\otimes\gamma\dot\omega)+
D\epsilon[\omega\gamma\otimes\gamma+\gamma\otimes\gamma\omega,\omega]\\
=&
[\mathbb I\omega,\omega]+D(1-\epsilon)[\omega\gamma\otimes\gamma+\gamma\otimes\gamma\omega,\omega]+D\epsilon[\omega\gamma\otimes\gamma+\gamma\otimes\gamma\omega,\omega]\\
=&[\mathbf k,\omega].
\end{align*}
\end{proof}

\begin{remark}\label{RAVAN}{\rm
If the radius $\sigma$ of the fixed sphere (the case (i)) tends to infinity, the parameter $\epsilon$ tends to 1, and the above equations reduce to the equations of the rolling of the Chaplygin ball over a horizontal hyperplane in $\R^n$ (see \cite{FeKo,Jo4}).
Also, note that the rolling of a Chaplygin ball over a sphere \eqref{sfera-sfera-red} is an example of a
modified L+R system on the product of $so(n)$ and the Stiefel variety $V_{n,r}$ for $r=1$, see Section 4.1 of \cite{Jo5}.
}\end{remark}

\begin{remark}
{\rm
Note that the mapping
\[
\xi \longmapsto (\xi\mathbf\Gamma)\wedge
\mathbf\Gamma=\xi\mathbf\Gamma\otimes\mathbf\Gamma+\mathbf\Gamma\otimes\mathbf\Gamma \xi
\]
is the
orthogonal projection $ \pr_\mathfrak v\colon so(n)\to\mathfrak v$ with respect to the scalar product
\eqref{KF}, while
$
\xi \longmapsto (\xi\gamma)\wedge
\gamma=\xi\gamma\otimes\gamma+\gamma\otimes\gamma \xi
$
is the orthogonal projection $ \pr_{\mathfrak v_\gamma}$ to $\mathfrak v_\gamma$, where the subspaces $\mathfrak v$ and $\mathfrak v_\gamma$ are defined by
\begin{equation}
\mathfrak v=\R^n\wedge\mathbf\Gamma \qquad \text{and} \qquad \mathfrak
v_\gamma=\Ad_{g^{-1}}(\mathfrak v)=\R^n \wedge \gamma.
\label{hg}\end{equation}

Then we have
\begin{equation}\label{epsilon}
\frac{d}{dt}\pr_{\mathfrak v_\gamma}=\epsilon[\pr_{\mathfrak
v_\gamma},\ad_\omega],
\end{equation}
where $[\cdot,\cdot]$ is the standard Lie bracket in the space of
linear operators of $so(n)$. Thus, equivalently, we can derive
\eqref{izvod} from the identity
$$
\frac{d}{dt}(\pr_{\mathfrak v_\gamma}\omega)=\pr_{\mathfrak
v_\gamma}\dot\omega+\frac{d}{dt}(\pr_{\mathfrak v_\gamma})\omega.
$$
}\end{remark}

\begin{remark}{\rm
The operator $\kappa=\mathbb I+D\pr_{\mathfrak v_\gamma}\colon
so(n)\to so(n)\cong so(n)^*$ can be also defined by the use of the \emph{constrained
Lagrangian}
\begin{equation}
\mathbf L=L\vert_{\mathbf V=\pm \rho \Ad_g(\omega)\mathbf\Gamma}= \frac12\langle \mathbb
I\omega,\omega\rangle+\frac{D}{2}(\Ad_g(\omega)\mathbf\Gamma,\Ad_g(\omega)\mathbf\Gamma)
=: \frac12\langle\kappa(\omega),\omega\rangle,
\label{red:lag}\end{equation}
which represents the kinetic energy, preserved along the flow of the system.
}\end{remark}

\subsection{Invariant measure}

Based on general observations given for $\epsilon$-modified L+R systems (see
Theorems 4 and 5, \cite{Jo5}) we have that for the rolling over a sphere,
the density of an invariant measure keeps the same form as in the case of the rolling over a horizontal hyperplane
(see Fedorov and Kozlov \cite{FeKo, FeRCD}).

Let
\begin{equation}
\mu(\gamma)=\sqrt{\det (\kappa)}=\sqrt{\det ({\mathbb I}+D\pr_{\mathfrak
v_\gamma})} \,, \label{ch:mer}
\end{equation}
and let $A=\diag(a_1,\dots,a_n)$, where $a_1,\dots,a_n$ are parameters of the inertia operator \eqref{ch-op}.
Also, by $\mathrm d\mathbf k$ and $\mathrm d\gamma$ we denote the standard volume
forms on $so(n)^*$ and $S^{n-1}$, respectively, and by $\varOmega$ the canonical symplectic structure on
$T^*SO(n)$, $d=\dim SO(n)$.

\begin{thm}\label{TC}
{\rm(i)} The problem of the rolling of a ball over a sphere
\eqref{sfera-sfera}  on $T^*SO(n)\times S^{n-1}$ in variables
$(\mathbf k,g,\gamma)$ has an invariant measure
\begin{equation}
\mu^{-1}\, \varOmega^d\wedge\mathrm d\gamma=1/\sqrt{\det(\kappa)}\,\varOmega^d\wedge
\mathrm d\gamma= 1/\sqrt{\det ({\mathbb I}+D\pr_{\mathfrak
v_\gamma})}\,\varOmega^d \wedge \mathrm d\gamma, \label{MU}
\end{equation}
while the reduced flow \eqref{sfera-sfera-red} in variables
$(\mathbf k, \gamma)$ has an invariant measure
\begin{equation}
\mu^{-1}\,\mathrm d\mathbf k \wedge \mathrm d\gamma= 1/\sqrt{\det
({\mathbb I}+D\pr_{\mathfrak v_\gamma})}\,\mathrm d\mathbf k
\wedge \mathrm d\gamma. \label{mu*}
\end{equation}

{\rm (ii)}
For the inertia operator \eqref{ch-op}, the density \eqref{ch:mer} is
proportional to
\[
(\gamma, A^{-1}\gamma)^{\frac{1}{2}({n-2})}.
\]
\end{thm}

\begin{remark}\label{PRIM}{\rm Since $\mathrm d\mathbf k=\det (\kappa)\mathrm d\omega$, the invariant measure of the reduced system considered
in variables $(\omega,\gamma)$ is $\mu(\gamma)\mathrm d\omega\wedge\mathrm d\gamma$.
}\end{remark}

\subsection{3--dimensional case}
In the case $n=3$, under the isomorphism between ${\mathbb R}^3$
and $so(3)$
\begin{equation}
\vec X=(X_1,X_2,X_3)\longmapsto X=\left(\begin{matrix}
0 & -X_3 & X_2 \\
X_3 & 0 & -X_1 \\
-X_2 & X_1 & 0
\end{matrix}\right), \label{iso}
\end{equation}
from \eqref{sfera-sfera-red},  we obtain the classical equations
of rolling without
slipping of the Chaplygin ball over a sphere
\begin{equation}
\frac{d}{dt}{\vec{\mathbf k}}=\vec{\mathbf k}\times\vec \omega,
\qquad \frac{d}{dt}{\vec \gamma}=\epsilon \vec
\gamma\times\vec\omega, \label{Chap}\end{equation} where
$\vec{\mathbf k}=\mathbb I \vec\omega+ D\vec \omega-D (\vec
\omega,\vec\gamma)\vec \gamma$ and $\mathbb I=\diag(I_1,I_2,I_3)$
is the inertia operator of the ball. In the space $\R^6(\vec{\omega},
\vec\gamma)$ the density \eqref{ch:mer}
 of an invariant measure is equal to
\begin{equation}
\mu(\vec\gamma)=\sqrt{ \det(\mathbb I+D\mathbb E)\big(1-D(\vec\gamma,(\mathbb
I+D\mathbb E)^{-1}\vec\gamma )\big)}, \label{mu-ch}
\end{equation} the expression given by Chaplygin for $\epsilon=1$ \cite{Ch1} (see Remark \ref{RAVAN}), and by Yaroshchuk for $\epsilon \ne 1$ \cite{Ya}. Here $\mathbb
E=\diag(1,1,1)$.

The system \eqref{Chap} always has three integrals
\begin{equation}
F_1=(\vec\gamma,\vec\gamma)=1, \quad F_2=\frac12(\vec{\mathbf
k},\vec \omega), \quad F_3=(\vec{\mathbf k},\vec{\mathbf k}).
 \label{cl:int}
\end{equation}

For $\epsilon=1$, there is the fourth integral $F_4=(\vec{\mathbf
k},\vec \gamma)$ and the problem is integrable by the Euler-Jacobi
theorem: the phase space is almost everywhere foliated by
two-dimensional invariant tori with quasi-periodic, non-uniform motion
\eqref{Jacobi2} (see Chaplygin \cite{Ch1}). Moreover,
Borisov and Mamaev proved that the system \eqref{Chap} is
Hamiltonizable with respect to certain nonlinear Poisson bracket
on $\R^6$ (\cite{BM}, see also \cite{BBM3, Ts}).

Remarkably, for $\epsilon=-1$ (the case (iii) with $\rho=2\sigma$)
Borisov and Fedorov (see \cite{BF}) found the
integrable case with the fourth integral
\[
\tilde F_4=(I_2+I_3-I_1+D)\mathbf
k_1\gamma_1+(I_3+I_1-I_2+D)\mathbf
k_2\gamma_2+(I_1+I_2-I_3+D)\mathbf k_3\gamma_3.
\]
The system is
integrated on an invariant hypersurface $\tilde F_4=0$ \cite{BFM}. Furthermore, its topological analysis
and a representation as a sum of two conformally Hamiltonian vector fields are given in \cite{BM3} and \cite{Ts}, respectively.
We feel that it would be very interesting to have similar results in a dimension greater then 3.

\section{Rolling of the Chaplygin ball without slipping and twisting}

\subsection{Rubber rolling}
Three--dimensional rubber Chaplygin ball problems are introduced
in \cite{EKR} and \cite{EK}, while the multidimensional rubber rolling over a horizontal
hyperplane is considered in \cite{Jo3}. For a given normal vector $\gamma=g^{-1}\mathbf\Gamma$, let
\[
\mathbf E_1,\dots,\mathbf E_{n-1},\mathbf\Gamma, \qquad \text{and}\qquad
\mathbf e_1=g^{-1} \mathbf E_1,\dots,\mathbf e_{n-1}=g^{-1} \mathbf E_{n-1},\gamma=g^{-1}\mathbf\Gamma
\]
be orhonormal bases of $\mathbb R^n$ in the fixed frame
and in the body frame, respectively.
Rubber Chaplygin ball is defined as a system \eqref{ch-constr}, \eqref{ch-lagr}
subjected
 to the additional constraints
 \begin{equation}\label{r-veze}
\phi_{ij}=\langle \Omega, \mathbf E_i\wedge \mathbf E_j\rangle=\langle \omega,\mathbf e_i\wedge \mathbf e_j\rangle=0, \qquad 1\le
i<j\le n-1
\end{equation}
describing the no-twist condition: the angular velocity matrix $\omega$ has rank 2 and the
corresponding admissible plane of rotation contains the normal vector $\gamma$ to the rolling sphere at the contact
point.

Alternatively, note that
\[
\mathbf E_i\wedge \mathbf E_j, \qquad \mathbf e_i\wedge \mathbf e_j=\Ad_g^{-1}(\mathbf E_i\wedge \mathbf E_j),
\qquad 1\le i<j\le n-1
\]
are the orthonormal bases of  $\mathfrak h$ and $\mathfrak
h_\gamma=\Ad_{g^{-1}}\mathfrak h$, orthogonal complements to $\mathfrak v$ and $\mathfrak v_\gamma$ (see \eqref{hg})
with respect to the scalar product \eqref{KF}.
Thus, the constraints \eqref{r-veze} can be rewritten as
\begin{equation}\label{rc}
\pr_\mathfrak h\Omega=0, \quad \text{i.e.,}\quad \pr_{\mathfrak h_\gamma}\omega=0 \quad \Longleftrightarrow \quad \Omega\in\mathfrak v, \quad \text{i.e.,} \quad \omega\in\mathfrak v_\gamma.
\end{equation}

As a
result, we obtain $(n-1)$-dimensional constraint distribution
\begin{equation}\label{DIS}
\mathcal F^\pm=\{(\omega,\mathbf V,g,\mathbf r) \, \vert\,
\mathbf V=\pm \frac{\rho}{\sigma\pm\rho} \Ad_g(\omega)\mathbf r, \, \pr_{\mathfrak
h_\gamma}\omega=0\}\subset\mathcal D^\pm.
\end{equation}

Let $\mathbb E$ be the identity operator on $so(n)$.
We have the relation
\begin{equation}\label{kI}
\mathbf k=\mathbb I\omega+D\omega=\mathbf I\omega,  \quad \text{for} \quad \omega\in\mathfrak v_\gamma=\R^n\wedge\gamma,
\end{equation}
where $\mathbf k$ is given by \eqref{K} and $\mathbf I=\mathbb I+\mathbb E$.
Let $m=\mathbf I\omega\in so(n)\cong so(n)^*$ be the angular momentum with
respect to the modified inertia operator $\mathbf I$.
After the identification $\mathcal D^{\pm}\cong TSO(n)\times S^{n-1}$,  we obtain a natural phase space of the problem:
\[
\mathcal G=\{(m,g,\gamma)\in T^*SO(n)\times S^{n-1}\,\vert\, \pr_{\mathfrak h_\gamma}\mathbf I^{-1}m=\pr_{\mathfrak h_\gamma}\omega=0\}.
\]

Using Proposition \ref{SSR} and \eqref{kI}, we can write the equations of
a motion in the variables $(m,g,\gamma)$
\begin{equation}
\dot m=[m,\omega]+\lambda_0, \qquad \dot g=g\cdot\omega, \qquad
\dot\gamma=-\epsilon\omega\gamma.\label{sfera-sfera-r2}
\end{equation}

The Lagrange multiplier $\lambda_0\in\mathfrak h_\gamma$
 is determined from the condition that the angular velocity $\omega$
satisfies \eqref{rc}.
From \eqref{epsilon} and the identity $\pr_{\mathfrak
h_\gamma}+\pr_{\mathfrak v_\gamma}=\mathbb E$, we have
$
\frac{d}{dt}\pr_{\mathfrak h_\gamma}=\epsilon[\pr_{\mathfrak
h_\gamma},\ad_\omega].
$
Thus,
\begin{align*}
0=& \frac{d}{dt}\big(\pr_{\mathfrak h_\gamma} \omega\big)\\
 =&\epsilon(\pr_{\mathfrak h_\gamma}\ad_\omega-\ad_\omega\pr_{\mathfrak
h_\gamma})\omega+\pr_{\mathfrak h_\gamma}\dot\omega\\
=& \pr_{\mathfrak h_\gamma}\frac{d}{dt}\big(\mathbf I^{-1}[m,\omega]+\mathbf I^{-1}\lambda_0 \big),
\end{align*}
and the multiplier $\lambda_0\in\mathfrak h_\gamma$ is the solution of the equation
\begin{equation}\label{EQL}
\mathbf I^{-1}([m,\omega]+\lambda_0)-\gamma\otimes\gamma
\mathbf I^{-1}([m,\omega]+\lambda_0)-\mathbf
I^{-1}([m,\omega]+\lambda_0)\gamma\otimes\gamma=0.
\end{equation}

Thus, we obtain.

\begin{prop}\label{SSRR}
The equations of a motion of the rubber Chaplygin ball on $\mathcal G$ are given by
\eqref{sfera-sfera-r2}, where $m=\mathbf I\omega=\mathbb I\omega+D\omega$, and $\lambda_0\in\mathfrak h_\gamma$ is the solution of \eqref{EQL}.
The reduction of the left $SO(n)$--symmetry \eqref{left2} induces
a system
on  the space $\mathcal G_0=\mathcal G/SO(n)=\{(m,\gamma)\in so(n)^*\times S^{n-1}\,\vert\, \pr_{\mathfrak h_\gamma}\omega=0\}$
given by \eqref{PE*} and
\begin{equation}
\dot m=[m,\omega]+\lambda_0. \label{sfera-sfera-r3}
\end{equation}
\end{prop}

The proof of the next theorem follows from considerations given in Subsection \ref{extended-system} below.

\begin{thm}\label{IMRR}
The problem of the rubber rolling of a ball
over a sphere \eqref{sfera-sfera-r2} and the reduced system \eqref{PE*}, \eqref{sfera-sfera-r3} possess
invariant measures
\[
\mu_\epsilon(\gamma)\,\varOmega^d\wedge \mathrm d\gamma\vert_{\mathcal G},\qquad
 \mu_\epsilon(\gamma)\,\mathrm dm\wedge\mathrm d\gamma\vert_{\mathcal G_0},
\]
respectively, where the density $\mu_\epsilon(\gamma)$ is given by
\begin{equation}\label{gustina}
\mu_\epsilon(\gamma)=(\det \mathbf I^{-1}\vert_{\mathfrak
h_\gamma})^\frac{1}{2\epsilon}\, .
\end{equation}
\end{thm}

\begin{remark}{\rm
Since $\mathrm dm=\det(\mathbf I)\mathrm d\omega=const\cdot\mathrm d\omega$, contrary to
remark \ref{PRIM}, here the reduced system considered in variables $(\omega,\gamma)$ has the invariant measure with the same density as in the variables $(m,\gamma)$:
$\mu_\epsilon\,\mathrm d\omega\wedge\mathrm d\gamma$.
}\end{remark}

\subsection{3--dimensional case}
For $n=3$, under the isomorphism \eqref{iso} between
$\R^3$ and $so(3)$  and the identification of $\vec\gamma$ with $\vec{\mathbf e}_1 \wedge \vec{\mathbf e}_2$ in \eqref{r-veze},
we have
\begin{equation}\label{e3}
\mathcal G_0=\{(\vec m,\vec\gamma)\in \R^3\times S^2\,\vert\, \phi=(\vec\gamma,\vec\omega)=0\}
\end{equation}
and the reduced system \eqref{sfera-sfera-r3}, \eqref{PE*}
reads
\begin{equation}\label{e3a}
\dot{\vec m}=\vec m\times \vec \omega+\lambda\vec\gamma, \qquad
\dot{\vec\gamma}=\epsilon \vec\gamma\times\vec\omega,
\end{equation}
where
\begin{equation}\label{e3b}
\vec m =(\mathbb I+D\mathbb E)\vec\omega=\mathbf I\vec\omega,
\qquad \lambda=-(\vec m,\mathbf I^{-1}(\vec m\times \vec \omega))/(\vec\gamma,\mathbf
I^{-1}\vec\gamma).
\end{equation}

The density \eqref{gustina} reduces to the well known expression
\begin{equation}\label{e3c}
\mu_\epsilon(\vec\gamma)=(\mathbf
I^{-1}\vec\gamma,\vec\gamma)^\frac{1}{2\epsilon}
\end{equation}
(see \cite{EKR} for $\epsilon=1$ and \cite{EK} for $\epsilon \ne 1$).
Apart of the integrability of the rolling over a
horizontal plane ($\epsilon=1$) \cite{EKR}, as in the case of
non-rubber rolling, Borisov and Mamaev  proved the
integrability for $\epsilon=-1$ \cite{BM2}. Note that for $\epsilon=1$, the above equations coincide with the equations of nonholonomic rigid body motion studied by Veselov and Veselova \cite{VeVe1, VeVe2}.

The problem is Haminltonizable for all $\epsilon$
\cite{EKR, EK}. On the other hand, the rubber rolling of the ball where the mass
center does not coincide with the geometrical center over a
horizontal plane provides an example of the system having the
following interesting property (see \cite{BBM2, BMB}). The appropriate
phase space is foliated on invariant tori, such that the foliation
is isomorphic to the foliation of integrable Euler case of the
rigid body motion about a fixed point, but the system itself has
not analytic invariant measure and is not Hamiltonizable.

\subsection{Extended system and a dual expression for an invariant measure}\label{extended-system}
Note that we can consider equations \eqref{e3a}, \eqref{e3b} on the
product $\R^3\times S^2$ as well. The system also has an invariant measure with density \eqref{e3c} and
the reduced system on \eqref{e3} is its subsystem ($\phi=(\vec \omega,\vec\gamma)$ is the first integral).
Similarly, the system \eqref{sfera-sfera-r3}, \eqref{PE*} can be extended and the invariant measure given in Theorem \ref{IMRR} is the restriction to  $\mathcal G_0$ of an invariant measure of the extended system.
In order to define the extended system such that we can
use the results of \cite{Jo5}, we need
to add some additional variables.

Firstly, consider the system \eqref{sfera-sfera-r3}, \eqref{PE*} on $\mathcal G_0$.
We can choose vectors $\mathbf e_i(t)$, $i=1,\dots,n-1$ along a trajectory
$(m(t),\gamma(t))$, such that $\mathbf e_1(t),\dots,\mathbf e_{n-1}(t),\mathbf e_n(t)=\gamma(t)$ is a orthonormal base of $\R^n$
and that
\begin{equation}\label{Ei}
\dot{\mathbf e}_i=-\epsilon \omega \mathbf e_i, \qquad
i=1,\dots,n.
\end{equation}
Indeed, we can take a base $\mathbf e_1(t_0),\dots,\mathbf e_{n}(t_0)$
at some initial time $t_0$ (it is defined modulo the orthogonal transformations of the hyperplane $\gamma(t_0)^\perp$). From
the modified Poisson equations \eqref{Ei} it follows that
the scalar products $(\mathbf e_i(t), \mathbf e_j(t))$  are conserved.

Further, the equations \eqref{Ei} imply
\begin{equation}\label{E1}
(\mathbf e_i\wedge \mathbf e_j)^{\mathbf\cdot}=\epsilon [\mathbf e_i\wedge \mathbf e_j,\omega], \qquad
1\le i<j\le n.
\end{equation}
We can determine the reaction force $\lambda_0$ starting
from the expression
\begin{equation}\label{L0}
\lambda_0=\sum_{1\le i<j\le n-1}\lambda^{ij} \mathbf e_i\wedge \mathbf e_j.
\end{equation}
and differentiating the constraints \eqref{r-veze} by using \eqref{sfera-sfera-r3} and \eqref{E1}. We get the Lagrange multipliers
$\lambda^{ij}$ in the form
\begin{equation}\label{lambde}
\lambda^{ij}=-\sum_{1\le k<l\le n-1} \langle \mathbf e_{k}\wedge
\mathbf e_l,\mathbf I^{-1}[m,\omega]\rangle \mathbf A^{ij,kl},
\end{equation}
where $\mathbf A^{ij,kl}$ is the inverse of the matrix $\mathbf
A_{ij,kl}=\langle \mathbf e_i\wedge \mathbf e_j,\mathbf I^{-1} \mathbf e_k\wedge
\mathbf e_l\rangle$.

The \emph{extended system} system on
\[
\mathcal M=\{(m,\mathbf e_1,\dots,\mathbf e_n)\,\vert \, m\in so(n)^*, \mathbf e_i\in \R^n, \, (\mathbf e_i,\mathbf e_j)=\delta_{ij}, \, 1\le i,j\le n\},
\]
is defined by the equation \eqref{sfera-sfera-r3} together with
\eqref{Ei}, \eqref{L0}, \eqref{lambde},
and the functions
\begin{equation}\label{geo-int}
\phi_{ij}=\langle \omega,\mathbf e_i\wedge \mathbf e_j\rangle, \qquad  1\le i<j\le n-1
\end{equation}
are its first integrals.

On the other hand, let
$
\mathcal N=so(n)^*\times\prod_{1\le i<j\le n}\mathcal O(\mathbf e_i\wedge\mathbf e_j),
$
where $\mathcal O(\mathbf e_i\wedge \mathbf e_j)$ is the adjoint orbit of $\mathbf e_i\wedge \mathbf e_j$ in $so(n)$.
The closed system defined by \eqref{sfera-sfera-r3}, \eqref{E1}, \eqref{L0}, \eqref{lambde}  on $\mathcal N$
is an example of a $\epsilon$-modified LR system introduced in \cite{Jo5}.
Now, the functions \eqref{geo-int} and
$\psi_{ij,kl}=\langle \mathbf e_i\wedge \mathbf e_j, \mathbf e_k\wedge \mathbf e_l\rangle$, $1\le i<j\le n,\,1\le k<l\le n$
are its first integrals. Also, the system has an invariant measure  (see Theorem 1, \cite{Jo5}):
\[
\mu_\epsilon\,\mathrm dm \bigwedge_{1\le i< j\le n}  \mathrm d(\mathbf e_i\wedge \mathbf e_j)\vert_\mathcal N,
\]
where
\begin{equation}\label{extended}
\mu_\epsilon=(\det \mathbf A_{ij,kl})^\frac{1}{2\epsilon}\qquad (1\le i<j\le n-1, \, 1\le k<l\le n-1).
\end{equation}

It easily follows that the extended system has an invariant measure
\[
\mu_\epsilon\,\mathrm dm\wedge \mathrm d\mathbf e_1 \wedge \dots \mathbf e_n\vert_{\mathcal M}
\]
(the replacing of equations \eqref{E1} by \eqref{Ei} do not reflect essentially on the corresponding Liouville equation).
Note that  the density \eqref{extended} of the extended system coincides with \eqref{gustina} and
the above statement implies invariant measures of the equations \eqref{sfera-sfera-r2} and \eqref{PE*}, \eqref{sfera-sfera-r3} given in Theorem \ref{IMRR}.

Next, by introducing the momentum
\begin{align}\label{MOMENT}
\mathbf m &= \pr_{\mathfrak v_\gamma} \mathbf
I\omega+\pr_{\mathfrak h_\gamma}\omega\\
&=\omega+ \gamma\otimes\gamma\,(\mathbf I\omega-\omega)+(\mathbf
I\omega-\omega)\,\gamma\otimes\gamma\in so(n)\cong
so(n)^*\nonumber,
\end{align}
we can describe the extended system  without using additional variables $\mathbf e_i$, $i=1,\dots,n-1$.
We have the \emph{momentum equation} (see \cite{Jo5}, i.e, \cite{FeJo} for $\epsilon=1$)
\begin{equation}\label{MOMENT*}
\dot{\mathbf m}=\epsilon [\mathbf
m,\omega]+(1-\epsilon)\pr_{\mathfrak v_\gamma}[\mathbf I\omega,\omega].
\end{equation}
Thus, we obtain an alternative description of the extended system on $so(n)^*\times S^{n-1}$ given by  \eqref{PE*} and \eqref{MOMENT*}.
It leads to the dual expression for an invariant
measure (see Theorems 2 and 4, \cite{Jo5}).

Let $A=\diag(a_1,\dots,a_n)$, where $a_1,\dots,a_n$ are parameters of the special inertia operator \eqref{spec-op}.

\begin{thm}\label{RS2}
{\rm (i)} The extended system \eqref{PE*}, \eqref{MOMENT*} of the rubber rolling of a ball over a fixed sphere in variables $(\mathbf m,\gamma)$ has an invariant measure
$\tilde \mu_\epsilon\,{\mathrm d}{\mathbf m}\wedge {\mathrm
d}\gamma$,
\begin{equation}\label{tilde mu}
\tilde{\mu}_\epsilon(\gamma)=(\det\mathbf I\vert_{\mathfrak v_\gamma})^{\frac{1}{2\epsilon}-1}.
\end{equation}

{\rm (ii)} For  $\mathbb I$ defined by \eqref{spec-op}, i.e., $\mathbf I(E_i\wedge E_j)=a_ia_j E_i\wedge E_j$,
the density \eqref{tilde mu} is proportional to
$$
(\mathbf \gamma, A \gamma)^{(\frac{1}{2\epsilon}-1)({n-2})}.
$$
\end{thm}

It is also clear that the momentum equation \eqref{MOMENT*}, together with $\dot g=g\omega$ and \eqref{PE*}, defines extended system on $T^*SO(n)\times S^{n-1}$
with an invariant measure $\tilde \mu_\epsilon\,\varOmega^d\wedge \mathrm d\gamma$.

\section{Reduction of $SO(n)$--symmetry}

\subsection{Chaplygin reduction to $TS^{n-1}$} As we already mentioned, the problem of the rubber rolling of a ball over a fixed sphere
is a $SO(n)$-Chaplygin system with respect to the action
\eqref{left}. We have the principal bundle \eqref{principal}, \eqref{submersion},
together with the principal connection
\begin{align}\label{connection}
 T_{(g,\mathbf r)} SO(n)\times S^{n-1}&  ={\mathcal F^\pm}_{(g,\mathbf
r)}\oplus \ker d\pi_{(g,\mathbf r)}, \\
 \ker d\pi_{(g,\mathbf r)} & =so(n)\cdot (g,\mathbf r).\nonumber
\end{align}

The system reduces to the tangent bundle $TS^{n-1}\cong \mathcal
F^\pm/SO(n)$.
The procedure of reduction for rubber rolling over a sphere for
$n=3$ is given by Ehlers and Koiller \cite{EK}. Note that in this
case the system is always Hamintonizable due to the fact that it
has an invariant measure and that the reduced configuration space
is 2--dimensional. We proceed with a reduction of $n$--dimensional
variant of the problem.

Recall that the vector in ${\mathcal F^\pm}_{(g,\mathbf r)}$ are
called \emph{horizontal}, while the vectors in $\ker
d\pi_{(g,\mathbf r)}$ {\it vertical}. The \emph{horizontal lift}
$\dot \gamma^h$ of the base vector $\dot\gamma\in T_\gamma
S^{n-1}$ to the horizontal space ${\mathcal F^\pm}$ at the point
$(g,\mathbf r)\in\pi^{-1}(\gamma)$ is the unique vector in
$\mathcal F^\pm_{(g,\mathbf r)}$ satisfying
$d\pi(\gamma^h)=\dot\gamma$.

\begin{lem}\label{lift}
The reduced Lagrangian on $TS^{n-1}=\mathcal F^\pm/SO(n)$ reads
\begin{align*}
L_{red}(\dot\gamma,\gamma)= \frac{1}{2\epsilon^2}\langle \mathbf
I(\gamma\wedge\dot\gamma),\gamma\wedge\dot\gamma\rangle=-\frac1{4\epsilon^2}\tr(\mathbf
I(\gamma\wedge\dot\gamma)\gamma\wedge\dot\gamma).
\end{align*}
\end{lem}

\begin{proof}
The horizontal lift $\dot\gamma^h\vert_{(g,\mathbf
r)}=(\omega,\mathbf V)$ is given by:
\begin{eqnarray*}
&& \omega=\frac{1}{\epsilon}\gamma\wedge\dot\gamma=\frac{\sigma\pm\rho}{\sigma}\gamma\wedge\dot\gamma,\\
&& \mathbf V=\dot{\mathbf r}=(\sigma\pm
\rho)\frac{d}{dt}(g\gamma)=(\sigma\pm \rho) (\dot
g\gamma+g\dot\gamma)=(\sigma\pm \rho)
(g\frac{1}{\epsilon}(\gamma\wedge\dot\gamma)
\gamma+g\dot\gamma)\\
&&\quad=(\sigma\pm
\rho)\big(1-\frac{1}\epsilon\big)g\dot\gamma=-(\sigma\pm
\rho)\big(\pm\frac\rho\sigma\big) g\dot\gamma.
\end{eqnarray*}

As a result, the reduced Lagrangian is
\begin{align*}
L_{red}(\dot\gamma,\gamma) = L(\dot\gamma^h\vert_{(g,\mathbf
r)},g,\mathbf r))_{(g,\mathbf
r)\in\pi^{-1}(\gamma)}
=\frac{1}{2\epsilon^2}\langle \mathbb
I(\gamma\wedge\dot\gamma),\gamma\wedge\dot\gamma\rangle+\frac{D}{2\epsilon^2}
(\dot\gamma,\dot\gamma),
\end{align*}
which proves the statement.
\end{proof}

The reduced Lagrange--d'Alembert equation describing the motion of
the system on a sphere $S^{n-1}$ takes the form
\begin{equation}\label{RedEq}
\Big( \frac{\partial L_{red}}{\partial \gamma} - \frac{d}{dt}
\frac{\partial L_{red}}{\partial \dot \gamma}, \xi\Big) =
\langle J_{(g,\mathbf r)}(\dot\gamma^h), K_{(g,\mathbf
r)}(\dot\gamma^h,\xi^h) \rangle, \qquad  \xi\in
T_\gamma S^{n-1},
\end{equation}
where $(g,\mathbf r)\in \pi^{-1}(\gamma)$, $K(\cdot,\cdot)$ is
$so(n)$--valued curvature of the connection, and $J$ is the
momentum mapping of $SO(n)$--action \eqref{left} (see \cite{Koi,
BKMM}).

It is well known that the momentum mapping
$$
J: T(SO(n)\times S^{n-1})\to so(n)\cong so(n)^*
$$
of the action \eqref{left} is given by
$$
J_{(g,\mathbf r)}(\omega,\mathbf V)=\Ad_g(\mathbb I\omega)+m\mathbf V
\wedge \mathbf r.
$$
Therefore,
\begin{align*} J_{(g,\mathbf r)}(\dot\gamma^h) &=\frac{1}\epsilon \Ad_g \mathbb I(\gamma\wedge\dot\gamma)- \frac{m}{\epsilon}
(\sigma\pm\rho)\big(\pm\frac\rho\sigma\big) g\dot\gamma\wedge \mathbf r\\
&=
\Ad_g\Big(\frac{1}\epsilon \mathbb I(\gamma\wedge\dot\gamma)\pm m
(\sigma\pm \rho)^2\frac\rho\sigma(
\gamma\wedge\dot\gamma ) \Big)\\
&=\frac{1}\epsilon\Ad_g\Big( \mathbb
I(\gamma\wedge\dot\gamma)\pm D
\frac{\sigma\pm \rho}{\rho}( \gamma\wedge\dot\gamma )\Big)\\
&= \frac{1}\epsilon\Ad_g\Big( \mathbb I(\gamma\wedge\dot\gamma)+
\frac{D}{1-\epsilon}( \gamma\wedge\dot\gamma )\Big).
\end{align*}

Let $\xi_1,\xi_2\in \mathcal F^\pm_{(g,\mathbf r)}$. By definition, the curvature
$K_{(g,\mathbf r)}(\xi_1,\xi_2)$ is the element $\eta\in so(n)$,
such that $\eta\cdot (g,\mathbf r)$ is the vertical component of
the commutator of vector fields $[X_2,X_1]$ at $(g,\mathbf r)$,
where $X_1$ and $X_2$ are smooth horizontal extensions of $\xi_1$
and $\xi_2$.

\begin{lem}\label{curvature} Let $\xi_1,\xi_2\in T_\gamma S^{n-1}$ and $(g,\mathbf
r)\in\pi^{-1}(\gamma)$.  Then
$$
K_{(g,\mathbf
r)}(\xi_1^h,\xi_2^h)=(1-\frac{\rho^2}{\sigma^2})\Ad_g(\xi_2\wedge\xi_1)=
\frac{2\epsilon-1}{\epsilon^2}\Ad_g(\xi_2\wedge\xi_1).
$$
In particular, for $\epsilon=1/2$, i.e, $\rho=\sigma$, the
curvature vanish and the constraints are holonomic.
\end{lem}

\begin{remark}{\rm
Note that the factor $1-\frac{\rho^2}{\sigma^2}$ equals to
$1-K_1/K_2$ where $K_1$ and $K_2$ are curvatures of the fixed and
rolling sphere, respectively.  The same factor appears in the case
of rubber rolling of arbitrary two surfaces in $\R^3$ (see
\cite{BH}).
}\end{remark}

Since $\langle
\gamma\wedge\dot\gamma,\dot\gamma\wedge\xi\rangle=0$, we can
replace $J$ by $\frac{1}\epsilon\Ad_g\left( \mathbf
I(\gamma\wedge\dot\gamma)\right)$ at the right hand side of \eqref{RedEq}, and we get the J-K term in the form
\begin{eqnarray*}
&&\langle J_{(g,\mathbf r)}(\dot\gamma^h), K_{(g,\mathbf
r)}(\dot\gamma^h,\xi^h)
\rangle=\frac{2\epsilon-1}{\epsilon^3}\langle\mathbf
I(\gamma\wedge\dot\gamma),\xi\wedge\dot\gamma\rangle\\
&&\qquad=-\frac{2\epsilon-1}{2\epsilon^3}\tr(\mathbf
I(\gamma\wedge\dot\gamma)\cdot
(\xi\otimes\dot\gamma-\dot\gamma\otimes\xi))=\frac{2\epsilon-1}{\epsilon^3}(\mathbf
I(\gamma\wedge\dot\gamma)\dot\gamma,\xi).
\end{eqnarray*}

We have
$$
\frac{\partial L_{red}}{\partial
\gamma}=\frac{1}{\epsilon^2}\mathbf
I(\gamma\wedge\dot\gamma)\dot\gamma, \qquad \frac{\partial
L_{red}}{\partial \dot \gamma}=-\frac{1}{\epsilon^2}\mathbf
I(\gamma\wedge\dot\gamma)\gamma.
$$

Therefore, we obtain the following statement.

\begin{thm}\label{REDsym}
The Lagrange--d'Alembert equation describing the motion of the
reduced system are given by
\begin{equation}\label{REDUCED}
\Big(\epsilon\frac{d}{dt}\big(\mathbf
I(\gamma\wedge\dot\gamma)\gamma\big)+(1-\epsilon)\mathbf
I(\gamma\wedge\dot\gamma)\dot\gamma,\xi\Big) =0 , \qquad \xi\in
T_\gamma S^{n-1}.
\end{equation}
\end{thm}

The above reduction slightly differs from the Chaplygin
$SO(n-1)$--reduction of the Veselova problem  studied in \cite{FeJo}.

\medskip

{\sc Proof of Lemma \ref{curvature}.}
In the coordinates $(g,\gamma)$, the $SO(n)$-action
takes the form \eqref{left2}. Let $\eta\in so(n)$. The associated
vector field on $SO(n)\times S^{n-1}$ with respect to the action
\eqref{left2} is given by
$$
\eta\cdot (g,\gamma)\cong (\Ad_{g^{-1}}\eta,0)\in T_{(g,\gamma)}
SO(n)\times S^{n-1},
$$
where, as above, we use the left trivialization of $TSO(n)$.
Further, the horizontal and vertical components of the vector
$(\omega,\xi)\in T_{(g,\gamma)}(SO(n)\times S^{n-1})$,
respectively, simply read
\begin{eqnarray*}
&& (\omega,\xi)^H=(\frac{1}{\epsilon}\gamma\wedge\xi,\xi), \\
&& (\omega,\xi)^V=(\omega-\frac{1}{\epsilon}\gamma\wedge\xi,0).
\end{eqnarray*}

Now, let $\xi_1,\xi_2$ be vector fields, the extensions of
$\xi_1,\xi_2\in T_{\gamma_0} S^{n-1}$ defined  in a neighborhood
$U$ of $\gamma_0$, and $X_1, X_2$ their horizontal lifts to
$SO(n)\times U$:
$$
X_i(g,\gamma)=(\frac{1}{\epsilon}\gamma\wedge\xi_i,\xi_i)=Y_i+Z_i,
\quad Y_i=(\frac{1}{\epsilon}\gamma\wedge\xi_i,0), \quad
Z_i=(0,\xi_i), \quad i=1,2.
$$

Then, by definition
$$
\langle
K_{(g,\gamma_0)}(\xi_1^h,\xi_2^h),\eta\rangle=\langle-[X_1,X_2]\vert_{(g,\gamma_0)}^V,\Ad_{g^{-1}}\eta\rangle,
$$
i.e.,
\begin{equation}\label{curv}
K_{(g,\gamma_0)}(\xi_1^h,\xi_2^h)=-\Ad_{g}[X_1,X_2]\vert_{(g,\gamma_0)}^V.
\end{equation}

We shall prove
\begin{equation}\label{curv2}
[X_1,X_2]^V=[X_1,X_2]=\big(\frac{2}{\epsilon}-\frac{1}{\epsilon^2}\big)(\xi_1\wedge\xi_2,0),
\end{equation}
which, according to \eqref{curv}, proves the lemma.

Without loosing a generality we may suppose that
$\gamma_0=(0,0,\dots,0,1)^T$. Let $(q_1,\dots,q_{n-1})\in U$ be
the local coordinates on the upper half-sphere
$S_+^{n-1}=\{\gamma\in S^{n-1}\,\vert\, \gamma_n>0\}$ defined by
\begin{eqnarray*}
&&\gamma_i=q_i, \qquad \quad i=1,\dots,n-1,\\
&&\gamma_n=\sqrt{1-q_1^2-\dots-q_{n-1}^2},\\
&& U=\{(q_1,\dots,q_{n-1})\in\R^{n-1}\,\vert\,
q_1^2+\dots+q_{n-1}^2<1\}.
\end{eqnarray*}
The given vectors $\xi_1,\xi_2\in T_{\gamma_0}S^{n-1}$ have the
form $(\xi_i^1,\dots,\xi_i^{n-1},0)^T$, $i=1,2$. By taking
$\xi_i^j=const$,
\begin{equation}\label{VF}
\xi_i=\sum_{j=1}^{n-1} \xi_i^j \frac{\partial}{\partial q_j}, \qquad i=1,2,
\end{equation}
define their natural commutative extensions to $U$. Note that
$\partial/\partial q_i$ corresponds to the vector field
$$
E_i-\frac{q_i}{\sqrt{1-q_1^2-\dots-q_{n-1}^2}}E_n=E_i-\frac{\gamma_i}{\gamma_n}E_n
$$
in redundant variables on $S_+^{n-1}\subset\R^n$, where we consider \eqref{baza} as vector fields on $\R^n$. Whence, in
redundant variables the
vector fields \eqref{VF} read
$$
\xi_i=(\xi^1_i,\dots,\xi^{n-1}_i,\xi_i^n)^T=(\xi^1_i,\dots,\xi^{n-1}_i,-\frac{1}{\gamma_n}(\xi_i^1\gamma_1+\dots+\xi_i^{n-1}\gamma_{n-1}))^T,
$$
$i=1,2$, implying the identities
\begin{align}
&\xi_i(\gamma_j)=\xi_i(q_j)=\xi^j_i,  \qquad
j=1,\dots,n-1, \nonumber\\
&\xi_i(\gamma_n)=\xi_i(\sqrt{1-q_1^2-\dots-q_{n-1}^2})=-\frac{\xi_i^1q_1+\dots+\xi_i^{n-1}
q_{n-1}}{\sqrt{1-q_1^2-\dots-q_{n-1}^2}}=\xi_i^n.\label{izvodi}
\end{align}

Let $E_{kl}=(E_k\wedge E_l,0)$. Then $Y_i=\sum_{k<l} y_i^{kl}
E_{kl}$, where
$$
y_i^{kl}=\frac{1}{\epsilon}\big((\gamma,E_k)(\xi_i,E_l)-(\gamma,E_l)(\xi_i,E_k)\big)=\frac{1}{\epsilon}\big(\gamma_k\xi_i^l-\gamma_l\xi_i^k\big),
$$
Next, due to the relations
$$
[E_{ij},E_{kl}]=([E_i\wedge
E_j,E_k\wedge E_l],0), \quad [E_{kl}, Z_i]=0, \quad [Z_1,Z_2]=0,
$$
on
$SO(n)\times U$, we get:
\begin{align}
\label{curv3} [X_1,X_2] =& \sum_{k<l,i<j}
[y_{1}^{kl}E_{kl},y_{2}^{ij}E_{ij}]+\sum_{i<j}[y_{1}^{ij}
E_{ij},Z_2]+\sum_{k<l}[Z_1,y_{2}^{kl} E_{kl}]\\
\nonumber =& (\frac{1}{\epsilon^2}[\gamma\wedge
\xi_1,\gamma\wedge\xi_2],0)+\sum_{i<j}
y_{1}^{ij}[E_{ij},Z_2]-\sum_{i<j} Z_2(y_{1}^{ij})E_{ij}\\
\nonumber & \quad \qquad
+\sum_{k<l}y_{2}^{kl}[Z_1,E_{kl}]+\sum_{k<l}
Z_1(y_{2}^{kl})E_{kl}\\
\nonumber  =& -\frac 1{\epsilon^2}([\xi_1,\xi_2],0)+\sum_{i<j}
(\xi_1(y_{2}^{ij})-\xi_2(y_{1}^{ij}))E_{ij}
\end{align}

On the other hand, from \eqref{izvodi} we obtain
\begin{align*}
\xi_1(y_{2}^{ij})-\xi_2(y_{1}^{ij}) =&
\frac{1}{\epsilon}\xi_1\big(\gamma_i\xi_2^j-\gamma_j\xi_2^i\big)-
\frac{1}{\epsilon}\xi_2\big(\gamma_i\xi_1^j-\gamma_j\xi_1^i\big)\\
=& \frac{1}\epsilon\big(\xi_1^i\xi_2^j-\xi_1^j\xi_2^i\big)-\frac{1}\epsilon\big(\xi_2^i\xi_1^j-\xi_2^j\xi_1^i\big)\\
=&\frac{2}{\epsilon}\big((\xi_1,E_i)(\xi_2,E_j)-(\xi_1,E_j)(\xi_2,E_i)
  \big),
\end{align*}
which together with \eqref{curv3} implies the relation
\eqref{curv2}.
\hfill $\Box$

\subsection{The reduced system on $T^*S^{n-1}$.}
Consider the Legendre transformation
\begin{equation}\label{Leg}
p=\frac{\partial L_{red}}{\partial \dot
\gamma}=-\frac{1}{\epsilon^2}\mathbf
I(\gamma\wedge\dot\gamma)\gamma.
\end{equation}
The point $(p,\gamma)$ belongs to the cotangent bundle of a sphere
realized as a symplectic submanifold in the symplectic linear
space $(\R^{2n}(p,\gamma),dp_1\wedge
dq_1+\cdots+dp_n\wedge dq_n)$:
\begin{equation}\label{psi}
(\gamma,\gamma)=1, \qquad (\gamma,p)=0.
\end{equation}

Let $\dot\gamma=\dot\gamma(p,\gamma)$
be the inverse of the Legendre transformation and
$$
\Upsilon=\Upsilon(\gamma,p)=\frac{1}{\epsilon^2}\left(\mathbf
I\left(\gamma\wedge\dot\gamma\right)\right)\dot\gamma\vert_{\dot\gamma=\dot\gamma(p,\gamma)}.
$$

Then we can write the equations \eqref{REDUCED} in the form
\begin{equation*}\label{REDUCED*}
\left(-\epsilon \dot p+(1-\epsilon)\Upsilon,\xi\right) =0 , \qquad
\xi\in T_\gamma S^{n-1},
\end{equation*}
which is equivalent either to
\begin{equation}\label{REDUCED**}
\epsilon \gamma \wedge \dot p+(\epsilon-1)\gamma\wedge \Upsilon=0,
\end{equation}
or to
\begin{equation}\label{dot P}
\dot
p=\frac{(1-\epsilon)}{\epsilon}\Upsilon+\mu\gamma,
\end{equation}
where the multiplier $\mu$ is determined from the equation
$$
\frac{d}{dt}(\gamma,p)=(\dot\gamma,p)+(\dot
p,\gamma)=(\dot\gamma,p)+\frac{(1-\epsilon)}{\epsilon}(\Upsilon,\gamma)+\mu(\gamma,\gamma)=0.
$$

\begin{prop}
The reduced flow on on the cotangent bundle $T^*S^{n-1}$ realized with constraints
\eqref{psi} takes the following form
\begin{equation}\label{dot G}
\dot\gamma =X_\gamma(p,\gamma), \qquad
\dot p =X_p(p,\gamma),
\end{equation}
where $X_\gamma$ is the inverse of the Legendre transformation \eqref{Leg} and
$$
X_p=\frac{(1-\epsilon)}{\epsilon^3}\left(\mathbf
I\left(\gamma\wedge X_\gamma\right)\right)X_\gamma+\Big(\frac{(\epsilon-1)}{\epsilon^3}(\left(\mathbf
I\left(\gamma\wedge X_\gamma\right)\right)X_\gamma,\gamma)-(X_\gamma,p)\Big)\gamma.
$$
\end{prop}

\subsection{The momentum equation and an invariant measure}
Alternatively, the reduced equation \eqref{REDUCED**} can be derived by using the momentum equation \eqref{MOMENT*}.
After the reduction to the sphere $S^{n-1}$, we
obtain
\begin{eqnarray*}
&&\mathbf m=\frac{1}{\epsilon}\left(\mathbf
I(\gamma\wedge\dot\gamma)\gamma\right)\wedge \gamma=\epsilon
\gamma\wedge p,\qquad \dot{\mathbf m}=\epsilon \dot\gamma\wedge p+\epsilon
\gamma\wedge \dot p,\\
&& \omega=\frac{1}{\epsilon}\gamma\wedge\dot\gamma,\qquad\qquad\qquad\,\, [\mathbf m,\omega]=[\gamma\wedge
p,\gamma\wedge\dot\gamma]=\dot\gamma \wedge p, \\
&& [\mathbf I\omega,\omega]= \frac{1}{\epsilon^2}\left(\mathbf
I(\gamma\wedge\dot\gamma)\gamma\wedge\dot\gamma-\gamma\wedge\dot\gamma
\mathbf I(\gamma\wedge\dot\gamma)\right)=\dot\gamma\wedge
p+\gamma\wedge \Upsilon.
\end{eqnarray*}

By putting those expressions into \eqref{MOMENT*} we
get the equation \eqref{REDUCED**}:
\begin{align*}
&\epsilon \dot\gamma\wedge p+\epsilon \gamma\wedge \dot p=
(2\epsilon-1)\dot\gamma \wedge p+(1-\epsilon)(\dot\gamma\wedge
p+\gamma\wedge \Upsilon)\\
\Longleftrightarrow \quad & \epsilon \gamma\wedge \dot p+(\epsilon-1)(\gamma\wedge \Upsilon)=0.
\end{align*}

As a bi-product, we get the following statement.

\begin{thm}\label{RMIM}
The reduced equations \eqref{dot G} has an invariant measure
$$
(\det\mathbf I\vert_{\mathfrak v_\gamma})^{\frac{1}{2\epsilon}-1}\mathbf w^{n-1},
$$
where $\mathbf w$ is the canonical symplectic form
\begin{equation}\label{csf}
\mathbf w=dp_1\wedge d\gamma_1+\dots+dp_n\wedge
d\gamma_n\,\vert_{T^*S^{n-1}}
\end{equation}
\end{thm}

\begin{proof}
The mapping
\[
\Phi\colon (\gamma,p)\mapsto (\gamma,\mathbf m), \qquad \mathbf m=\epsilon
\gamma\wedge p,
\]
together with $\omega=\frac{1}{\epsilon}\gamma\wedge\dot\gamma$,
maps the reduced system \eqref{dot G} to the subsystem of
\eqref{PE*}, \eqref{MOMENT*}, and the pull-back $\Phi^*(\mathrm d\mathbf m\wedge\mathrm d\gamma)$ is
the standard volume form $\mathbf w^{n-1}$ on $T^*S^{n-1}$ (up to the multiplication by a constant).
Now the statement follows from  Theorem \ref{RS2}, item (i).
\end{proof}

\section{Hamiltonization of the reduced system}

\subsection{Equations for the special inertia operator}
Based on the Hamiltoniazation and integrability of the reduced Veselova system \cite{FeJo},
we have the Hamiltonization and
integrability of the rubber rolling of a Chaplygin ball over a
horizontal hyperplane for a special inertia operator \eqref{spec-op} (see \cite{Jo3}). Namely, under the time substitution $d\tau
=1/\sqrt{(A\gamma,\gamma)}dt$,
the reduced system becomes an integrable Hamiltonian system
describing a geodesic flow on $S^{n-1}$
of the metric
\begin{equation} \label{dsA}
ds^2_A=\frac{1}{(\gamma, A\gamma)}\left((A
{d\gamma},{d\gamma})(A\gamma,\gamma) - (A\gamma,{d\gamma})^2\right),
\end{equation}
where $d\gamma=(d\gamma_1,\dots,d\gamma_n)^T$ \cite{Jo3}.

Now we proceed with a rolling over a sphere and, as in the case of the horizontal rolling, we suppose that the inertia operator is given by \eqref{spec-op}. Then the reduced
Lagrangian $L_{red}(\dot \gamma,\gamma)$ and the Legendre transformation
\eqref{Leg} take the form
\begin{eqnarray}
&& L_{red} = \frac 1{2\epsilon^2} \left((A\dot \gamma,\dot
\gamma)(A\gamma,\gamma)- (A\gamma,\dot \gamma)^2 \right) .\label{redL}\\
&& \label{moments} p= \frac{\partial L_{red}}{\partial\dot \gamma}
=\frac{1}{\epsilon^2} (\gamma, A \gamma) A \dot \gamma -
\frac{1}{\epsilon^2}(\dot \gamma, A \gamma) A \gamma .
\end{eqnarray}

Under conditions (\ref{psi}), relations (\ref{moments}) can be
uniquely inverted to yield
\begin{equation}
\label{extended1} \dot
\gamma=\frac{\epsilon^2}{(\gamma,A\gamma)}
\left( A^{-1}p - (p,A^{-1}\gamma) \gamma\right)\,
\end{equation}
implying that the angular velocity in terms of $(p,\gamma)$ takes the form
$$
\omega(p,\gamma)=\frac{1}{\epsilon}\gamma\wedge \dot
\gamma=\frac{\epsilon}{(\gamma,A\gamma)} \gamma\wedge A^{-1}p,
$$
and we get:
\begin{align*}
\Upsilon(p,\gamma) &=\frac{1}{\epsilon^2}\left(A\gamma\wedge
A\dot\gamma\right)\dot\gamma=\frac{1}{\epsilon^2}\left((A\dot\gamma,\dot\gamma)A\gamma-(A\gamma,\dot\gamma)A\dot\gamma\right)\\
=&\frac{\epsilon^2}{(\gamma,A\gamma)^2}(p - (p,A^{-1}\gamma)
A\gamma, A^{-1}p - (p,A^{-1}\gamma)
\gamma)A\gamma\\
&-\frac{\epsilon^2}{(\gamma,A\gamma)^2}(\gamma,p -
(p,A^{-1}\gamma) A\gamma)\left[p - (p,A^{-1}\gamma)
A\gamma\right]\\
=&\frac{\epsilon^2}{(\gamma,A\gamma)^2}\left(
(A^{-1}p,p)+(A\gamma,\gamma)(p,A^{-1}\gamma)^2 \right)A\gamma\\
&+\frac{\epsilon^2}{(\gamma,A\gamma)^2}(p,A^{-1}\gamma)
(A\gamma,\gamma)\left[p - (p,A^{-1}\gamma)
A\gamma\right],
\end{align*}
that is
$$
\Upsilon(p,\gamma)=\frac{\epsilon^2}{(\gamma,A\gamma)^2}\left(
(A^{-1}p,p)A\gamma+(p,A^{-1}\gamma) (A\gamma,\gamma)p\right).
$$

In particular,
$(\Upsilon(p,\gamma),\gamma)={\epsilon^2}(A^{-1}p,p)/{(\gamma,A\gamma)}$, and
the right hand side of equation \eqref{dot P} reads
\begin{align*}
X_p(p,\gamma)  =& \frac{(1-\epsilon)}{\epsilon}\Upsilon+\frac{(\epsilon-1)}{\epsilon}(\Upsilon,\gamma)\gamma-(\dot\gamma,p)\gamma\\
=& \frac{(1-\epsilon)\epsilon}{(\gamma,A\gamma)^2}\left(
(A^{-1}p,p)A\gamma+(p,A^{-1}\gamma)
(A\gamma,\gamma)p\right)+\frac{\epsilon(\epsilon-1)}{(\gamma,A\gamma)}(A^{-1}p,p)\gamma\nonumber\\
&-\frac{\epsilon^2}{(\gamma,A\gamma)} \left( (A^{-1}p,p) -
(p,A^{-1}\gamma) (\gamma,p)\right)\gamma.
\end{align*}

Finally, we obtain the equation
\begin{equation}\label{extended2}
\dot p=\frac{\epsilon(1-\epsilon)}{(\gamma,A\gamma)^2}\left(
(A^{-1}p,p)A\gamma+(p,A^{-1}\gamma)
(A\gamma,\gamma)p\right)-\frac{\epsilon}{(\gamma,A\gamma)}(p,A^{-1}p)\gamma.
\end{equation}

By combing Theorems \ref{RS2} and \ref{RMIM}, we get.

\begin{thm}  The reduced flow of the rubber Chaplygin ball rolling over a sphere with a inertia operator \eqref{spec-op}
on the cotangent bundle $T^*S^{n-1}$ realized with constraints
\eqref{psi} is given by equations \eqref{extended1} and \eqref{extended2}.
The system has an invariant measure
\begin{equation}\label{RSIM}
(A\gamma,\gamma)^{\frac{n-2}{2\epsilon}+2-n}\mathbf w^{n-1}.
\end{equation}
\end{thm}

\subsection{The Chaplygin reducing multiplier.}
The Hamiltonian function of the reduced system takes
the form
\begin{equation}
H=\frac {\epsilon^2}{2}  \frac{ (p, A^{-1}p)}{(\gamma,A\gamma)},
\label{hamiltonian}
\end{equation}
which is unique only on the subvariety (\ref{psi}).

At the points of $T^*S^{n-1}$, the system \eqref{extended1}, \eqref{extended2} can be written in the
almost Hamiltonian form
\begin{equation}\label{AHF}
\dot x=X_H=(X_p,X_\gamma), \qquad i_{X_H}(\mathbf w+\Sigma)=dH,
\end{equation}
where $\Sigma$ is a semi-basic perturbation term, determined by the J-K term at the right hand side of \eqref{RedEq} (e.g,
see \cite{EKR, CCLM, St, Tat}). The form $\mathbf w+\Sigma$ is non-degenerate, but, in general, it is not closed.

The \emph{Chaplygin multiplier}
is a nonvanishing function $\nu$ such that $\tilde{\mathbf
w}=\nu(\mathbf w+\Sigma)$ is closed. The Hamiltonian vector field $\tilde X_H$ of the function $H$ on $(T^*S^{n-1},\tilde{\mathbf w})$
 is proportional to the original vector field:
\[
\tilde X_H=\frac{1}{\nu}X_H, \qquad
i_{\tilde
X_H}\tilde{\mathbf w}=dH.
\]
Thus, applying the time substitution $d\tau={\nu}dt$, the system \eqref{AHF} becomes the Hamiltonian system
\[
\frac{d}{d\tau}x=\tilde
X_H.
\]

On the other hand, a  classical way to introduce the Chaplygin
reducing multiplier for our system is as follows (e.g., see
\cite{Ch2, FeJo}). Consider the time substitution $d\tau={\nu}(\gamma)dt$, and denote $\gamma^{\prime}={d\gamma}/{d\tau}=\dot\gamma/\nu$. Then the Lagrangian function
transforms to $L^*(\gamma^{\prime},\gamma)=L_{red}(\nu\gamma^{\prime},\gamma)$ and we have the new momenta
$\tilde p={\partial L^*}/{\partial \gamma^\prime}=\nu p$.
The factor $\nu$ is Chaplygin reducing multiplier if under the above time reparameterization the equations
\eqref{extended1}, \eqref{extended2} become Hamiltonian in the
coordinates $(\tilde p, \gamma)$.

The existence of the  Chaplygin reducing multiplier $\nu$ implies that
the original system has
an invariant measure $\nu^{n-2}\mathbf w^{n-1}$ (e.g., see Theorem 3.5, \cite{FeJo}).
From the expression of an invariant measure \eqref{RSIM} we get the form of a possible Chaplygin multiplier:
\[
\nu(\gamma)=const \cdot (A\gamma,\gamma)^{\frac{1}{2\epsilon}-1}.
\]

Remarkably, we have.

\begin{thm}\label{main} Under the time substitution  $d\tau
=\epsilon(A\gamma,\gamma)^{\frac{1}{2\epsilon}-1}\, dt$ and an
appropriate change of momenta, the reduced system
\eqref{extended1}, \eqref{extended2} becomes a Hamiltonian system
describing a geodesic flow on $S^{n-1}$ with the metric
\begin{equation}
ds^2_{A,\epsilon}=(\gamma, A\gamma)^{\frac{1}{\epsilon}-2}
\left((A d\gamma, d\gamma)(A\gamma,\gamma) -(A\gamma, d\gamma)^2
\right). \label{dsAe}
\end{equation}
\end{thm}

\begin{remark}{\rm
Note that, while reductions and invariant measures of considered nonholonomic systems are given for arbitrary inertia tensors (Sections 2 and 3), the Hamiltonization is performed only for the special one \eqref{spec-op}. This assumption implies that $\mathbf I=\mathbb I+D\mathbb E$ preserves the subset of bivectors in $so(n)$. For $n\ge 4$, it is a restrictive property, while for $n=3$ an arbitrary inertia operator can be written in the form \eqref{spec-op} and we reobtain the result of Ehlers and Koiller \cite{EK}. This is expected since only if the reduced configuration space is two-dimensional, the existence of an invariant measure is equivalent to the existence of a Chaplygin multiplier (e.g., see \cite{BBM3}).
}\end{remark}

\begin{proof} We take
$\nu(\gamma)=\epsilon(A\gamma,\gamma)^{\frac{1}{2\epsilon}-1}$, so the
Lagrangian (\ref{redL}) in the new time becomes
\begin{equation} \label{L^*}
L^*(\gamma^{\prime},\gamma)=\frac12(\gamma,
A\gamma)^{\frac{1}{\epsilon}-2}\left((A \gamma^{\prime},
\gamma^{\prime})(A\gamma,\gamma) -(A\gamma, \gamma^{\prime})^2
\right).
\end{equation}

Following the method of Chaplygin reducing multiplier, we introduce the new momenta by
considering the mapping
\begin{equation}
(p,\gamma) \longmapsto  (\tilde p,\gamma),  \qquad \tilde
p=\nu p=\epsilon(A\gamma,\gamma)^{\frac{1}{2\epsilon}-1} p\,.
\label{mapping}
\end{equation}
Under (\ref{mapping}), the Hamiltonain (\ref{hamiltonian})
transforms to
\begin{equation}\label{Ham2}
H(\tilde p,\gamma)=\frac12 (\gamma,
A\gamma)^{1-\frac{1}{\epsilon}}{ (\tilde p, A^{-1}\tilde p)}.
\end{equation}

Now, we realize
the cotangent bundle $T^*S^{n-1}$  within $\R^{2n}(\tilde p,\gamma)$:
\begin{equation}
\psi_1= (\gamma,\gamma)=1, \quad \psi_2= (\tilde p,\gamma)=0,
\label{psi*}
\end{equation}
endowed with the symplectic structure
\[
\tilde{\mathbf w}=d\tilde p_1\wedge
dq_1+\cdots+d\tilde p_n\wedge dq_n\vert_{T^*S^{n-1}}.
\]

It is convenient to obtain the Hamiltonian vector field $\tilde
X_{H}=(\tilde X_{\tilde p},\tilde X_\gamma)$ of $H$ on
$(T^*S^{n-1},\tilde{\mathbf w})$ by using the Lagrange multipliers
(e.g., see \cite{AKN}). Let
\[\mathcal H=H-\lambda \psi_1-\mu
\psi_2.
\]
Then the equations of the geodesic flow of the metric
$ds^2_{A,\epsilon}$ can be written as
\begin{eqnarray*}
&& \gamma^\prime=\tilde X_\gamma=\frac{\partial \mathcal
H}{\partial \tilde p}= (\gamma,
A\gamma)^{1-\frac{1}{\epsilon}}A^{-1}\tilde p-\mu
\gamma,\\
&& \tilde p^\prime=\tilde X_{\tilde p}=-\frac{\partial \mathcal
H}{\partial \gamma}=\frac{1-\epsilon}{\epsilon}(\gamma,
A\gamma)^{-\frac{1}{\epsilon}}(\tilde p, A^{-1}\tilde
p)A\gamma+2\lambda \gamma+\mu \tilde p,
\end{eqnarray*}
where the multipliers $\lambda$ and $\mu$ are determined by taking
the derivative of the constraints \eqref{psi*}.
The straightforward calculations yield
\begin{eqnarray*}
&& \psi_1^\prime=2(\gamma,
A\gamma)^{1-\frac{1}{\epsilon}}(A^{-1}\tilde p,\gamma)-2\mu
(\gamma,\gamma)=0,\\
&& \psi_2^\prime=(\gamma,
A\gamma)^{1-\frac{1}{\epsilon}}(A^{-1}\tilde p,\tilde p)-\mu
(\gamma,\tilde p)\\
&&\qquad+\frac{1-\epsilon}{\epsilon}(\gamma,
A\gamma)^{-\frac{1}{\epsilon}}(\tilde p, A^{-1}\tilde
p)(A\gamma,\gamma)+2\lambda (\gamma,\gamma)+\mu (\tilde
p,\gamma)=0,
\end{eqnarray*}
implying
\begin{align*}
\mu =& (\gamma, A\gamma)^{1-\frac{1}{\epsilon}}(A^{-1}\tilde
p,\gamma), \\
 2\lambda =& -(\gamma, A\gamma)^{1-\frac{1}{\epsilon}}(A^{-1}\tilde
p,\tilde p)+\frac{\epsilon-1}{\epsilon}(\gamma,
A\gamma)^{1-\frac{1}{\epsilon}}(\tilde p, A^{-1}\tilde p) \\
=& -\frac{1}{\epsilon}(\gamma,
A\gamma)^{1-\frac{1}{\epsilon}}(\tilde p, A^{-1}\tilde p).
\end{align*}

Therefore, the Hamiltonian flow of $H$ on $(T^*S^{n-1},\tilde{\mathbf w})$ takes the
form
\begin{eqnarray}
&& \gamma^\prime=(\gamma,
A\gamma)^{1-\frac{1}{\epsilon}}\left(A^{-1}\tilde p-(A^{-1}\tilde
p,\gamma)
\gamma\right),\label{ham_q}\\
&& \tilde p^\prime=\frac{1-\epsilon}{\epsilon}(\gamma,
A\gamma)^{-\frac{1}{\epsilon}}(\tilde p, A^{-1}\tilde
p)A\gamma  \label{ham_p} \\
&&\qquad\qquad -\frac{1}{\epsilon}(\gamma,
A\gamma)^{1-\frac{1}{\epsilon}}(\tilde p, A^{-1}\tilde p)
\gamma+(\gamma, A\gamma)^{1-\frac{1}{\epsilon}}(A^{-1}\tilde
p,\gamma) \tilde p.\nonumber
\end{eqnarray}

In the time $t$, after inverting the mapping (\ref{mapping}), the
equation (\ref{ham_q}) takes the form
$$
\dot\gamma\cdot \frac1\epsilon
(A\gamma,\gamma)^{1-\frac{1}{2\epsilon}}=\epsilon(A\gamma,\gamma)^{\frac{1}{2\epsilon}-1}(\gamma,
A\gamma)^{1-\frac{1}{\epsilon}}\left(A^{-1}p-(A^{-1}p,\gamma)
\gamma\right),
$$
which coincides with (\ref{extended1}). Further, from
\begin{align*}
\frac{d}{d\tau}\tilde p =&
\frac{d}{d\tau}\Big(\epsilon(A\gamma,\gamma)^{\frac{1}{2\epsilon}-1}{p}\Big)=
\frac{d}{dt}\Big(\epsilon(A\gamma,\gamma)^{\frac{1}{2\epsilon}-1}p\Big)\frac1\epsilon
(A\gamma,\gamma)^{1-\frac{1}{2\epsilon}} \\
=& \Big(p\frac{d}{dt}(A\gamma,\gamma)^{\frac{1}{2\epsilon}-1}+\dot
p(A\gamma,\gamma)^{\frac{1}{2\epsilon}-1}\Big)(A\gamma,\gamma)^{1-\frac{1}{2\epsilon}},
\end{align*}
and
\begin{align*}
\frac{d}{dt}(A\gamma,\gamma)^{\frac{1}{2\epsilon}-1} =2\big(\frac{1}{2\epsilon}-1\big)(A\gamma,\gamma)^{\frac{1}{2\epsilon}-2}(A\gamma,\dot\gamma)
= (2\epsilon-1)\epsilon(A\gamma,\gamma)^{\frac{1}{2\epsilon}-2}
(p,A^{-1}\gamma),
\end{align*}
we get
\begin{equation}\label{*}
\frac{d}{d\tau}\tilde p=(2\epsilon-1)\epsilon(A\gamma,\gamma)^{-1}
(p,A^{-1}\gamma)p+\dot p.
\end{equation}

Finally, by combining \eqref{*} with the right hand side of
\eqref{ham_p} written in variables $(p,\gamma)$,
\begin{align*}
\tilde X_{\tilde p}(p,\gamma)=&\epsilon({1-\epsilon})(\gamma,
A\gamma)^{-2}(p, A^{-1}p)A\gamma\\
&-\epsilon(A\gamma,\gamma)^{-1}( p, A^{-1} p)
\gamma+\epsilon^2(A\gamma,\gamma)^{-1}(A^{-1} p,\gamma) p,
\end{align*}
we obtain the equation (\ref{extended2}):
\begin{align*}
\dot p =& (1-2\epsilon)\epsilon(A\gamma,\gamma)^{-1}
(p,A^{-1}\gamma)p+\epsilon({1-\epsilon})(\gamma, A\gamma)^{-2}(p,
A^{-1}p)A\gamma\\
&-\epsilon(A\gamma,\gamma)^{-1}( p, A^{-1} p)
\gamma+\epsilon^2(A\gamma,\gamma)^{-1}(A^{-1} p,\gamma) p\\
=& \frac{\epsilon(1-\epsilon)}{(\gamma,A\gamma)^2}\left(
(A^{-1}p,p)A\gamma+(p,A^{-1}\gamma)
(A\gamma,\gamma)p\right)-\frac{\epsilon}{(\gamma,A\gamma)}(p,A^{-1}p)\gamma.
\end{align*}

We proved that the vector field defining the motion is
proportional to the Hamiltonian vector field
$$
X_H=(X_p,X_\gamma)=\epsilon(A\gamma,\gamma)^{\frac{1}{2\epsilon}-1}(\tilde
X_{\tilde p},\tilde X_\gamma)=\epsilon(A\gamma,\gamma)^{\frac{1}{2\epsilon}-1}\tilde X_{H},
$$
and, whence, $\nu=\epsilon(A\gamma,\gamma)^{\frac{1}{2\epsilon}-1}$
is the Chaplygin multiplier of the system.
\end{proof}

\begin{remark}{\rm
For $\epsilon=1$, \eqref{dsAe} becomes the metric for the horizontal rolling \eqref{dsA}.
The geodesic flow of the metric \eqref{dsA} is
completely integrable \cite{FeJo}.
As in the 3-dimensional case, it is possible to prove the
complete integrability of the reduced systems for $\epsilon=-1$ and arbitrary $A$, as well as for $\epsilon\ne -1$ with matrixes $A$ having additional symmetries.
We shall consider the integrability aspects of the
problem and a geometrical setting by using nonholonomic connections following \cite{Ba, DrGa, Koi} in a separate paper.
}\end{remark}

\subsection*{Acknowledgments}
The author is very grateful to Yuri Fedorov, Borislav Gaji{\'c}, and the referees for many valuable suggestions that helps the author to improve the exposition
of the results.  The research was supported by the Serbian Ministry of
Science Project 174020, Geometry and Topology of Manifolds,
Classical Mechanics and Integrable Dynamical Systems.

\end{document}